\documentclass[sigconf]{acmart}

\usepackage{amsmath}
\usepackage{booktabs,url}
\usepackage[ruled,lined,boxed,linesnumbered]{algorithm2e}
\def \vd{{\mathbf d}}
\def \ve{{\mathbf e}}
\def \vf{{\mathbf f}}

\def \vg{{\mathbf g}}

\def \vr{{\mathbf r}}
\def \vs{{\mathbf s}}
\def \vu{{\mathbf u}}
\def \vv{{\mathbf v}}
\def \vw{{\mathbf w}}
\def \vx{{\mathbf x}}
\def \vy{{\mathbf y}}
\def \vz{{\mathbf z}}

\def \vB{{\mathbf B}}

\def \vD{{\mathbf D}}

\def \vW{{\mathbf W}}

\def \setA{\mathcal{A}}
\def \setB{\mathcal{B}}

\def \setP{\mathcal{P}}
\def \setG{\mathcal{G}}
\def \setH{\mathcal{H}}

\def \setS{\mathcal{S}}
\def \setV{\mathcal{V}}

\def \setE{\mathcal{E}}

\AtBeginDocument{%
  \providecommand\BibTeX{{%
    \normalfont B\kern-0.5em{\scshape i\kern-0.25em b}\kern-0.8em\TeX}}}

\copyrightyear{2021}
\acmYear{2021}
\setcopyright{acmcopyright}\acmConference[WSDM '21]{Proceedings of the
Fourteenth ACM International Conference on Web Search and Data Mining}{March
8--12, 2021}{Virtual Event, Israel}
\acmBooktitle{Proceedings of the Fourteenth ACM International Conference on Web
Search and Data Mining (WSDM '21), March 8--12, 2021, Virtual Event, Israel}
\acmPrice{15.00}
\acmDOI{10.1145/3437963.3441756}
\acmISBN{978-1-4503-8297-7/21/03}

\begin{document}
\fancyhead{}

\title{Exploring the Subgraph Density-Size Trade-off\\ via the Lov\'asz Extension}

\author{Aritra Konar}
\affiliation{%
  \institution{University of Virginia}
  \city{Charlottesville}
  \state{Virginia}
  \country{USA}
}
\email{aritra@virginia.edu}

\author{Nicholas D. Sidiropoulos}
\affiliation{%
  \institution{University of Virginia}
  \city{Charlottesville}
  \state{Virginia}
  \country{USA}
}
\email{nikos@virginia.edu}

\begin{abstract}
Given an undirected graph, the \textsc{Densest}-$k$-\textsc{Subgraph} problem (DkS) seeks to find a subset of $k$ vertices such that the sum of the edge weights in the corresponding subgraph is maximized. The problem is known to be NP-hard, and is also very difficult to approximate, in the worst-case. In this paper, we present a new convex relaxation for the problem. Our key idea is to reformulate DkS as minimizing a submodular function subject to a cardinality constraint. Exploiting the fact that submodular functions possess a convex, continuous extension (known as the Lov\'asz extension), we propose to minimize the Lov\'asz extension over the convex hull of the cardinality constraints. Although the Lov\'asz extension of a submodular function does not admit an analytical form in general, for DkS we show that it does. We leverage this result to develop a highly scalable algorithm based on the Alternating Direction Method of Multipliers (ADMM) for solving the relaxed problem. Coupled with a pair of fortuitously simple rounding schemes, we demonstrate that our approach outperforms existing baselines on real-world graphs and can yield high quality sub-optimal solutions which typically are \emph{a posteriori} no worse than $65-80\%$  of the optimal density.
\end{abstract}

\ccsdesc[500]{Mathematics of computing~Graph algorithms}
\ccsdesc[500]{Mathematics of computing~Continuous optimization}
\ccsdesc[500]{Mathematics of computing~Submodular optimization and polymatroids}

\keywords{Dense subgraphs; submodularity; Lov\'asz extension; convex optimization; Alternating Direction Method of Multipliers}

\maketitle

\section{Introduction}
\noindent \textbf{Motivation and Overview:}
Dense subgraph discovery is a key primitive in graph mining that finds application in diverse disciplines ranging from computational biology \cite{saha2010dense}, chemical informatics \cite{podolyan2009common}, network science \cite{chen2010dense,zhang2012extracting,giatsidis2014corecluster} and fraud detection \cite{hooi2016fraudar,zhang2017hidden}. Given an unweighted, undirected graph on $n$ vertices, the classical \textsc{DensestSubgraph} problem \cite{goldberg1984finding} seeks to determine the subgraph with the largest average degree. The problem can be solved exactly in polynomial-time and approximately (but quasi-optimally) via a greedy algorithm \cite{charikar2000greedy}. Recent work has extended these ideas to take into account higher-order structure in graphs \cite{tsourakakis2015k,mitzenmacher2015scalable}.

A drawback of the aforementioned approaches is that they do not feature a means of explicitly controlling the size of the desired subgraph. Hence, if one is interested in computing the densest subgraph as a function of the  size $k$ with the aim of exploring the optimal density-size trade-off, an additional cardinality constraint on the subgraph size has to be imposed in the formulation of \textsc{DensestSubgraph}. Unfortunately, this simple modification renders the resulting problem, known as \textsc{Densest}-$k$-\textsc{Subgraph} (DkS), NP-hard. Furthermore, the problem is known to be notoriously difficult to approximate in the worst-case \cite{khot2006ruling,bhaskara2012polynomial,manurangsi2017almost}. 

\noindent \textbf{Prior Art:} The state-of-the-art approximation algorithm \cite{bhaskara2010detecting} for the DkS problem provides a worst-case approximation guarantee of $O(n^{1/4 + \epsilon})$ (for some $\epsilon >0$) in time $n^{O(1/\epsilon)}$ for every choice of $k$, which is a very pessimistic result in general. Restricted cases of the problem are known to enjoy better approximation guarantees. For dense graphs, where the number of edges $m = \Omega(n^2)$ and for linear subgraph sizes $k = \Omega(n)$, a $1+\epsilon$ approximation algorithm was presented in \cite{arora1999polynomial}. However, the result has limited implications for real-world networks since they are sparse in edges (with $m = O(n)$) \cite{watts1998collective}. For general sizes $k$, a $O(n/k)$ approximation can be achieved by applying a greedy algorithm \cite{feige2001dense} or via semidefinite relaxation \cite{srivastav1998finding,feige2001approximation}. Note that in the linear size regime $k = \Omega(n)$, this yields a constant-factor approximation. That being said, in practice, for large graphs one is more interested in the \emph{sublinear} size regime $k = o(n)$, where the bounds again become very pessimistic. Recently, a new semidefinite relaxation approach for DkS has been proposed in \cite{bombina2020convex} that guarantees exact recovery in planted dense subgraph models with high probability. However, real world graphs are not known to obey such synthetic models. Additionally, the high complexity incurred in solving the semidefinite program
is a limitation of the approach. 

In a departure from such worst-case results, the recent work of \cite{papailiopoulos2014finding} approaches the problem via the lens of low-rank matrix factorization.
Specifically, it is shown that if the graph adjacency matrix has constant rank (in $n$), then the DkS problem is solvable in polynomial-time. When the adjacency matrix is not constant rank, solving the problem using low-rank approximation still yields an \emph{a posteriori} graph-dependent upper bound on the optimal density for a given $k$. Through experiments on real-world graphs, it is shown that the approach yields high-quality solutions that can come close to attaining the upper bound in certain cases.

\noindent \textbf{Approach and Contributions:} In this paper, we propose a new convex relaxation of the DkS problem for obtaining high-quality, sub-optimal solutions. 
Our contributions can be summarized as follows:
\begin{itemize}
    \item We reformulate the DkS problem as minimizing a \emph{submodular} function  subject to a cardinality constraint. Leveraging the fact that submodular functions are endowed with a convex, continuous extension (i.e., the Lov\'asz extension), we devise a new convex formulation for DkS that minimizes the Lov\'asz extension over the convex hull of the cardinality constraints.
    \item In general, the Lov\'asz extension of a submodular function does not admit an analytical form. In this case however, by judiciously exploiting the structure inherent in the problem, we establish a simple closed-form expression for the Lov\'asz extension. We utilize this result to develop an efficient and scalable algorithm for solving the convex relaxation via an inexact variant of the popular Alternating Direction Method of Multipliers (ADMM) \cite{lions1979splitting,boyd2011distributed}, which features computationally lightweight updates and guaranteed convergence.
    \item The solution of our relaxed problem is not guaranteed to be integral in general. Hence, we perform post-processing via two simple rounding schemes to obtain final integral solutions for DkS. While we do not possess \emph{a priori} guarantees on the quality of the obtained solution at present, via experiments on real-world graphs we demonstrate that our approach can consistently outperform prominent baselines. In fact, utilizing the upper bound on the optimal edge-density developed in \cite{papailiopoulos2014finding}, we demonstrate that \emph{a posteriori} our approach can discover dense subgraphs that are typically no worse than $65-80\%$ of the optimal density. 
\end{itemize}
On a final note, to put our contributions into perspective, we note that the prevailing approach to convex relaxation for combinatorial quadratic programming problems has been semidefinite relaxation \cite{luo2010semidefinite}, which is the Lagrangian bi-dual of the original problem, and hence is the closest convex problem to DkS, in a certain sense.  Since the Lov\'asz extension is the convex envelope of a submodular function \cite{lovasz1983submodular}, our convex relaxation can be viewed as an alternative which is the closest convex problem to DkS in a different sense (this notion is made precise in Section 4). 

\section{Primer on Submodularity}
We provide an overview of basic concepts regarding submodular functions \cite{lovasz1983submodular,fujishige2005submodular,bach2013learning}.
Given a set of $n$ objects $\setV = \{1,\cdots,n\}$, a \emph{set function} $F:2^{\mathcal{V}} \rightarrow \mathbb{R}$ assigns a real value to any subset $\mathcal{S} \subseteq \mathcal{V}$. 
\newline
\noindent\textbf{Definition 1. [Submodularity]} A set function $F(.)$ is said to be \emph{submodular} if and only if for all subsets $\mathcal{A},\mathcal{B} \subseteq \mathcal{V}$, it holds that
\begin{equation}\label{eq:SubModDef}
 F(\mathcal{A} \cup \mathcal{B}) + F(\mathcal{A} \cap \mathcal{B})
 \leq 
 F(\mathcal{A})+F(\mathcal{B}).
\end{equation}
The above definition can be equivalently, and more conveniently, restated in the following form.
\newline
\noindent\textbf{Definition 2.} For all $\mathcal{A}\subseteq\mathcal{B}\subseteq \mathcal{V}\setminus \{v\}$, it holds that
\begin{equation}\label{eq:SubModDef2}
F(\mathcal{A} \cup \{v\})-F(\mathcal{A}) 
\geq F(\mathcal{B} \cup \{v\})-F(\mathcal{B}).
\end{equation}
That is, for such functions, given subsets $\setA \subseteq \setB \subseteq \setV \setminus \{v\}$, the marginal improvement obtained by adding an element $v$ to the larger set  $\mathcal{B}$ never exceeds that obtained by adding $v$ to its subset $\mathcal{A}$.  Simply stated, equation \eqref{eq:SubModDef2} asserts that submodular functions exhibit a diminishing returns property.
\newline
\noindent\textbf{Definition 3. [The Lov\'asz extension]}
A remarkable feature of submodular functions is that they possess a continuous, \emph{convex} extension known as the Lov\'asz extension, which extends their domain from $2^{\setV}$ to the unit interval $[0,1]^n$ (recall $n=|\setV|$). Formally, the Lov\'asz extension $f_L:[0,1]^n \rightarrow \mathbb{R}$ of a submodular function $F(.)$ is defined as
\begin{equation}\label{eq:lovasz}
    f_L(\vx) := \underset{\vg \in \setB_F}{\max} \; \vg^T\vx,
\end{equation}
where the set $\setB_F$ is the \emph{base polytope} associated with $F(.)$ and is defined as 
\begin{equation}\label{eq:basePolytope}
    \setB_F:= \{\vg \in \mathbb{R}^n : \vg^T\mathbf{1}_n = F(\setV); \vg^T\mathbf{1}_{\setS} \leq F(\setS), \forall \; \setS  \subseteq \setV\}.
\end{equation}
From equation \eqref{eq:lovasz}, it is evident that the Lov\'asz extension corresponds to the support function of the base polytope $\setB_F$, and is thus convex. In fact, it can be shown that $f_L(.)$ is convex if and only if $F(.)$ is submodular. Furthermore, when evaluated at a binary vector $\vx \in \{0,1\}^n$, the Lov\'asz extension equals the value of the submodular function $F(.)$.

\section{Problem Statement}
In this section, we formally describe the \textsc{Densest}-$k$-\textsc{Subgraph} (DkS) problem.
Consider a weighted, undirected, simple graph  $\setG := (\setV,\setE,w)$ on $n$ vertices, with vertex set $\setV: = \{1,\cdots,n\}$ and edge set $\setE \subseteq \setV \times \setV$ consisting of $m:=|\setE|$ edges. The function $w:\setE \rightarrow \mathbb{R}_{++}$ assigns each edge with a positive weight, and we collect these weights in a vector $\vw \in \mathbb{R}^m_{++}$. In the special case that $\vw$ is the vector of all-ones, we say that the graph $\setG$ is unweighted.

Given a positive integer $1 < k < n$, we consider the problem of computing the subset of vertices $\setS \subset \setV$ of size $k$ such that the sum of the edge weights in the induced subgraph $\setG_{\setS}$ is as larges as possible. The DkS problem can be expressed in quadratic programming form as 
\begin{equation}\label{eq:DKS}
\begin{aligned}
& \underset{\vx \in \{0,1\}^n}{\text{max}}
& &\vx^T\vW\vx \\
& \quad \text{s.to}& & \mathbf{1}^T\vx = k,
\end{aligned}
\end{equation}
where $\vW$ represents the $n \times n$ (weighted) adjacency matrix of the graph $\setG$. Note that each binary vector $\vx \in \{0,1\}^n$  corresponds to the indicator vector of a vertex subset $\setS \subseteq \setV$ , i.e., we have
\begin{equation}
x_i = 
\begin{cases}
1, \;\text{if} \; i \in \setS \\
0, \;\text{otherwise}. 
\end{cases}
\end{equation}
Hence, for a given subset of vertices $\setS \subset \setV$, the objective function counts the total weight of the edges in the subgraph $\setG_{\setS}$ induced by $\setS$, while the constraints ensure that $\setS$ contains precisely $k$ vertices. 

Regarding computational complexity, problem \eqref{eq:DKS} is known to be NP--hard in its general form (it contains the \textsc{MaximumClique} problem as a special case \cite{feige2001dense}). Additionally, the problem also has a documented history of resistance to efficient approximation in polynomial-time \cite{khot2006ruling,bhaskara2012polynomial,manurangsi2017almost}. Notwithstanding such pessimistic worst-case results, in this paper we devise a new polynomial-time approximation algorithm for the DkS problem that relies on exploiting the combinatorial structure of \eqref{eq:DKS} in a principled manner. Our approach is outlined in the following section.

\section{Proposed Approach}

Consider the following equivalent reformulation of problem \eqref{eq:DKS} in subset selection form
\begin{equation}
\underset{|\setS| = k }{\min} \biggl\{F(\setS): -\mathbf{1}_{\setS}^T\vW\mathbf{1}_{\setS} \biggr\},
\end{equation}
where $\mathbf{1}_{\setS}$ denotes the binary indicator vector of subset $\setS \subset \setV$. We now make the following crucial observation regarding the cost function.
\begin{theorem}
The cost function $F(.)$ is submodular.
\end{theorem}
\begin{proof}
Note that for a given subset $\setS$, the cost function is linearly separable over the edge set $\setE_{\setS}$ of the induced subgraph $\setG_{\setS}$, i.e., we have
\begin{equation}
    F(\setS) = \sum_{(i,j) \in \setE_{\setS}}-w_{ij},
\end{equation}
where $w_{ij}$ denotes the weight of edge $(i,j) \in \setE_{\setS}$. Since submodularity is preserved under summation, in order to obtain the desired result, it suffices to show that each constituent function
\begin{equation}\label{eq:littleF}
F_{ij}(\setS) := 
\begin{cases}
-w_{ij}, & \;\text{if} \; i \;\text{and}\; j \in \setS, \\
\quad 0, & \;\text{otherwise}, 
\end{cases}
\end{equation}
is submodular. Defining the pair of sets $\setA := \setS \cap \{i\}, \setB := \setS \cap \{j\}$  and applying Definition $1$ then completes the proof.
\end{proof}

Although the above observation does not make the (NP--hard) DkS problem any easier to solve, it does open the door to the following approximation approach. First, we define the set 
\begin{equation}
 \setP: = \{ \vx \in [0,1]^n; \mathbf{1}^T\vx = k\}  
\end{equation}
to be the convex hull of the combinatorial sum-to-$k$ constraints. 
Since submodular functions are endowed with a convex, continuous extension (the Lov\'asz extension) which equals the value of $F(.)$ at all binary $\{0,1\}^n$ vectors, the DkS problem can be equivalently expressed as 
\begin{equation} \label{eq:DKS2}
 \begin{aligned}
& \quad {\text{min}}
& & f_L(\vx) \\
& \quad \text{s.to}& &\vx \in  \{0,1\}^n \cap \setP.
\end{aligned}   
\end{equation}
On dropping the combinatorial constraints, we obtain the relaxed problem
\begin{equation} \label{eq:Lrelax}
\underset{\vx \in \setP}{\text{min}}\;\; f_L(\vx)
\end{equation}
which we refer to as the Lov\'asz relaxation. Clearly, the above problem is convex, and can be solved in polynomial-time to obtain a lower bound on the optimal value of \eqref{eq:DKS2}.

We now outline our primary motivation for employing the Lov\'asz extension. Before proceeding, we recall a few basics of convex analysis \cite{rockafellar1970convex}. Given any function $f:\mathbb{R}^n \rightarrow \mathbb{R} \cup \{+\infty\} $, its \emph{Fenchel conjugate} is defined as $f^*(\vy) := \sup_{\vx} \{\vy^T\vx - f(\vx)\}$, which is always closed and  convex (even if $f(.)$ is not). Taking the conjugate of $f^*(.)$ yields the \emph{biconjugate} $f^{**}(.)$ of the function $f(.)$, which is also closed and convex, and an under-estimator of $f(.)$, i.e., $f^{**} \leq f$.  As a matter of fact, the biconjugate $f^{**}(.)$ constitutes the convex closure of $f(.)$, and thus, is the tightest convex under-estimator of $f(.)$ (in a certain sense). The link between the Lov\'asz extension of a submodular function and its Fenchel biconjugate is provided by the following result, which is extracted from \cite{lovasz1983submodular,bach2013learning}.
\begin{lemma}
Given a subodular function $F(.)$, define the function 
\begin{equation}
g(\vx): = 
	\begin{cases}
	F(\setS), &\forall\; \vx = \mathbf{1}_{\setS}, \setS \subseteq 2^{\setV}\\
	+\infty, &\forall\; \vx \neq \{0,1\}^n.
	\end{cases} 
\end{equation}
Then, the Fenchel biconjugate of $g(.)$ is the Lov\'asz extension of $F(.)$.
\end{lemma}
\noindent Hence, the Lov\'asz extension corresponds to the convex closure, or the tightest convex under-estimator (in the above sense) of the submodular function $F(.)$ on the domain $[0,1]^n$, which justifies its use as a principled, continuous relaxation of the quadratic cost function of DkS.  

While the above result places the Lov\'asz relaxation \eqref{eq:Lrelax} on a firm theoretical footing, from an algorithmic perspective, a notable drawback of the approach is that the Lov\'asz extension does not admit an analytical form in general. This stems from the fact that $f_L(.)$ is the support function of the base polytope $\setB_F$ of $F(.)$ (see equation \eqref{eq:lovasz}), which is characterized by (potentially) an exponential number of inequalities in the problem dimension $n$. In his seminal work \cite{edmonds1970submodular}, Edmonds presented a simple greedy algorithm for computing a subgradient of the Lov\'asz extension at any point $\vx \in [0,1]^n$ in time $O(n\log n +nT)$ \footnote{Here, $T > 0$ is an upper bound on the maximum time taken to evaluate $F(.)$ for any choice of subset $\setS \subseteq \setV$.} \emph{without} explicitly constructing $\setB_F$. While this fact can be exploited to solve the Lov\'asz relaxation \eqref{eq:Lrelax} via a projected subgradient algorithm, such an approach suffers from slow convergence. Indeed, the primal convergence rate (i.e., convergence to the optimal value) of subgradient methods for convex problems is $O(1/\sqrt{t})$ \cite{nesterov2013introductory}, where $t$ is the number of iterations. Hence, adopting such an approach is limited to producing low-accuracy solutions for large-scale problems. 

 We now demonstrate that it is possible to solve the Lov\'asz relaxation (hereafter referred to as the L-relaxation) in a substantially more efficient manner. Our key result is that for the DkS problem, we can explicitly characterize the base polytope of the submodular cost function $F$, which in turn allows us to obtain an analytical form for the Lov\'asz extension. Finally, we apply a primal-dual algorithm that leverages the explicit structure of the problem to compute efficient solutions for the L-relaxation.

Before proceeding, we introduce the following notation: let $\vd: = \vW\mathbf{1}_n$ represent the (weighted) degree vector of the vertices of $\setG$, and $\vB \in \{-1,0,1\}^{n \times m}$ denote the directed vertex-edge incidence matrix of the graph $\setG$. Note that a column of $\vB$ corresponds to an edge $(i,j) \in \setE$, and is of the form $(\ve_i-\ve_j)$, where $\ve_i$ denotes the $i^{th}$ canonical basis vector in $\mathbb{R}^n$. We are now ready to state our main result.

\begin{theorem}
The base polytope of $F$ can be expressed as 
$$ \mathcal{B}_F = \{\vg \in \mathbb{R}^n : \vg = -\vd+\vB\vf, \forall \, |\vf| \leq \vw \}.$$
\end{theorem}
\begin{proof}
Once again, we exploit the fact that the function $F$ can be linearly decomposed as $F(\setS) = \sum_{(i,j) \in \setE} F_{ij}(\setS)$, where $F_{ij}$ has been previously defined in  \eqref{eq:littleF}. An important result \cite[Theorem 44.6]{schrijver2003combinatorial} regarding such decomposable submodular functions asserts that the base polytope can be expressed as the set-addition of the base polytopes of the constituent functions  $\{F_{ij}\}_{(i,j) \in \setE}$, i.e., we have
$
\mathcal{B}_F = \sum_{(i,j) \in \setE} \setB_{F_{ij}},
$
where $\setB_{F_{ij}}$ is the base polytope of $F_{ij}$. This suggests that if we can find a simple expression for each constituent base polytope $\setB_{F_{ij}}$, then we can possibly characterize the full polytope $\setB_F$.

To this end, consider a component function $F_{ij}$. Applying the definition \eqref{eq:basePolytope}, its base polytope can be expressed as 
\begin{equation}
 \setB_{F_{ij}} = \{\vg \in \mathbb{R}^n: g_i \leq 0, g_j \leq 0, g_i + g_j \leq -w_{ij}, g_i + g_j = -w_{ij}\},    
\end{equation}
which in turn can be re-expressed as
\begin{equation}
\begin{aligned}
\setB_{F_{ij}} &= -w_{ij}\textrm{conv}(\ve_i,\ve_j), \\
&= -w_{ij}[\alpha_{ij}(\ve_i-\ve_j) + \ve_j],  \forall \; \alpha_{ij} \in [0,1].
\end{aligned}
\end{equation}
Introducing the change of variable $\beta_{ij} := 1-2\alpha_{ij}$, we obtain
\begin{equation}
\begin{aligned}
\setB_{F_{ij}} &= \frac{-w_{ij}}{2}[(\ve_i + \ve_j) - \beta_{ij}(\ve_i - \ve_j) ], \\
&= \frac{1}{2}[w_{ij}\beta_{ij}(\ve_i - \ve_j) - w_{ij}(\ve_i + \ve_j)], \forall \; \beta_{ij} \in [-1,1].
\end{aligned}
\end{equation}
This allows us to obtain the complete representation 
\begin{equation}
\setB_F = -\frac{1}{2} \sum_{(i,j) \in \setE} w_{ij}(\ve_i + \ve_j)
+ \frac{1}{2} \sum_{(i,j) \in \setE} w_{ij}\beta_{ij}(\ve_i - \ve_j).
\end{equation}
Note that the first summand is precisely the (weighted) degree vector $\vd$ (as the contribution of  each vertex $i \in \setV$ is $\sum_{j:(i,j) \in \setE} w_{ij}\ve_i$), while the second summand can be expressed as $\vB\vf$, where $\vf \in \mathbb{R}^m$ is a vector with entries $f_{ij}:= w_{ij}\beta_{ij}$. Since $|\beta_{ij}| \leq 1$, by construction, we have $|f_{ij}| \leq w_{ij}$, and thus $|\vf| \leq \vw$. Putting everything together, we finally obtain the following characterization of the base polytope
\begin{equation}
    \setB_F = \frac{1}{2}[-\vd + \vB\vf], \forall \; |\vf| \leq \vw,
\end{equation}
which yields the desired result up to the global scaling factor $1/2$. 
\end{proof}
As an immediate consequence of the above result, we obtain the following analytical form for the Lov\'asz extension.
\begin{corollary}
The Lov\'asz extension of $F$ is 
$$f_L(\vx) = -\vd^T\vx + \sum_{(i,j) \in \setE} w_{ij}|x_i - x_j|.$$
\end{corollary}
\begin{proof}
Utilizing the form of the base polytope, we can express the Lov\'asz extension as 
\begin{equation}
\begin{aligned}
f_L(\vx) &= \underset{|\vf| \leq \vw}{\max} (-\vd + \vB\vf)^T\vx \\
&= -\vd^T\vx + \underset{|\vf| \leq \vw}{\max} (\vB^T\vx)^T\vf \\
&= -\vd^T\vx + \sum_{(i,j) \in \setE} w_{ij}|x_i - x_j|,
\end{aligned}
\end{equation}
where in going from the second to the third step we have utilized the fact that the vector $\vB^T\vx$ generates pair-wise differences between entries of $\vx$ that are connected by an edge in $\setG$.   
\end{proof}
The above result allows us to express the L-relaxation (in maximization form) as
\begin{equation}\label{eq:Lrelax2}
    \underset{\vx \in \setP}{\max} \biggl\{\vd^T\vx - \sum_{(i,j) \in \setE} w_{ij}|x_i - x_j|\biggr\}.
\end{equation}
An intuitive explanation of the above formulation is as follows. For any binary vector $\vx \in \setP$ that represents an induced subgraph $\setG_S$, the first term in the above objective function is a measure of the volume of $\setG_S$, i.e., it is the sum of the degrees of all the vertices in the induced subgraph. Meanwhile, the second term, which corresponds to graph total variation, counts the weighted sum of all edges crossing the boundary of $\setG_S$, i.e., it measures the cut. The difference of these two terms is then (twice) the sum of all edges in $\setG_S$, which is precisely the objective function that the DkS problem seeks to maximize. Equivalently stated, we wish to find an induced subgraph on $k$ vertices with high volume and small cut.

When solving the L-relaxation, we allow for non-binary vectors $\vx \in \setP$. In this case, the value of each entry of $\vx$ is a soft ``membership'' score that reflects the ``likelihood'' of a vertex belonging to the $k$-densest subgraph. The objective function then places higher emphasis on those likelihood profiles where the membership values are largest for those vertices that have large degree and are simultaneously ``smooth'' (in the total-variation sense) with respect to their one-hop neighbors, which is an intuitive proxy for dense subgraphs of size $k$. 

Hence, the L-relaxation constitutes a meaningful relaxation of the DkS problem. That being said, the form of the Lov\'asz extension reveals that problem \eqref{eq:DKS2} is neither differentiable, nor strongly concave, which constitutes a computational impediment in solving it efficiently at scale. In the next section, we show that by exploiting the structure of the problem in an intelligent fashion, it is in fact possible to develop an efficient and scalable algorithm. 

\section{Algorithms}

In order to motivate our algorithmic approach, we express the L-relaxation in the following manner. First, we define the functions
\begin{equation}
    g(\vx):=
    \begin{cases}
    -\vd^T\vx, & \vx \in \setP,\\
    +\infty, & \textrm{otherwise}
    \end{cases}
\end{equation}
and
\begin{equation}
h(\vx):= \sum_{(i,j) \in \setE} w_{ij}|x_i - x_j|
= \|\vD\vB^T\vx\|_1,
\end{equation}
where $\vD:=\textrm{diag}(\vw)$ denotes a diagonal matrix containing the edge-weights $\vw$. We can now express problem \eqref{eq:Lrelax} as
\begin{equation}
    \underset{\vx \in \mathbb{R}^n}{\min} ~\; g(\vx) + h(\vx),
\end{equation}
which in turn is equivalent to 
\begin{equation} \label{eq:ADMM1}
 \begin{aligned}
& \underset{\vx, \vz \in \mathbb{R}^n}{\text{min}}
& & g(\vx) + h(\vz)\\
& \quad \text{s.to}& &\vx - \vz = \mathbf{0}.
\end{aligned}   
\end{equation}
The above problem is now in a form  suitable for the application of the Alternating Direction Method of Multipliers (ADMM) \cite{lions1979splitting,boyd2011distributed} - a flexible framework for solving convex optimization problems that fuses the benefits of dual decomposition and augmented Lagrangian techniques into a simple primal-dual algorithm. The main utility of ADMM is that it decomposes complicated cost functions into simpler components (these can be non-smooth or even represent embedded constraints) via variable splitting and allows them to be handled separately, while featuring guaranteed convergence to the optimal solution of the problem under very mild assumptions. While being a very general framework for solving convex optimization problems, ADMM is most efficient when its sub-problems admit an analytical or simple computational solution. 

For the particular form of variable splitting employed in problem \eqref{eq:ADMM1}, it can be shown that the ADMM updates are given by 
\begin{subequations}
\begin{align}
\vx^{t+1} &= \textrm{prox}_{\rho g}(\vz^t - \vu^t) \\
\vz^{t+1} &= \textrm{prox}_{\rho h}(\vx^{t+1} + \vu^t) \\
\vu^{t+1} &= \vu^t + \vx^{t+1} - \vz^{t+1}
\end{align}
\end{subequations}
where $\vu \in\mathbb{R}^n$ is the normalized dual variable associated with the consensus constraint, $\rho > 0$ is a tuning parameter, and 
\begin{equation} \label{eq:prox}
    \textrm{prox}_{\rho f} (\vv): = \arg \underset{\vx \in \mathbb{R}^n}{\min} \; f(\vx) + \frac{\rho}{2}\|\vx - \vv\|_2^2
\end{equation}
denotes the \emph{proximal operator} \cite{parikh2014proximal} of a closed, proper, convex function $f$. It has been shown \cite{he20121} that the algorithm converges at a rate of $O(1/t)$, which represents an order of magnitude improvement over subgradient methods. However, since ADMM accesses the functions $g,h$ via their proximal operators, the overall efficiency of the algorithm depends on the complexity of evaluating these operators.

First, we focus on the complexity of the $\vx$- update, i.e., computing the proximal operator of the function $g$. Our next result shows that it admits a simple solution. 
\begin{lemma}
The optimal solution $\vx^* := \textrm{prox}_{\rho g}(\vv)$ is characterized by the pair of conditions
$$ x_i^* = \max \biggl\{\min{\biggl(v_i + (1/\rho)(d_i - \nu^*),1 \biggr)},0\biggr\}, \forall \; i \in [n], 
\sum_{i=1}^n x_i^* = k,  
$$
where $\nu^* \in \mathbb{R}$ is the optimal dual variable associated with the sum-to-$k$ constraint.
\end{lemma}
\begin{proof}
Define the function
\begin{equation}
    \Tilde{f}(\vx):=
    \begin{cases}
    -\vd^T\vx + \frac{\rho}{2}\|\vx- \vv\|_2^2, & \mathbf{0} \leq \vx \leq \mathbf{1},\\
    + \infty, & \textrm{otherwise}. 
    \end{cases}
\end{equation}
Then, the proximal operator of $g$ is given by
\begin{equation}\label{eq:proxg2}
    \textrm{prox}_{\rho g}(\vv) = \arg \underset{\mathbf{1}^T\vx = k}{\min} \; \Tilde{f}(\vx).
\end{equation}
The Lagrangian of the above problem is 
\begin{equation}
L(\vx,\nu) := 
\begin{cases}
-\vd^T\vx + (\rho/2)\|\vx-\vv\|_2^2 + \nu(\mathbf{1}^T\vx - k), & \mathbf{0} \leq \vx \leq \mathbf{1},\\
+ \infty, & \textrm{otherwise}
\end{cases}
\end{equation}
where $\nu \in \mathbb{R}$ is the dual variable associated with the equality constraint. Let $(\vx^*,\nu^*)$ denote the primal-dual optimal pair of \eqref{eq:proxg2}. The Karush-Kuhn-Tucker (KKT) conditions (which are necessary and sufficient for optimality in this case) assert that the pair $(\vx^*,\nu^*)$ satisfy
\begin{subequations}
\begin{align}
\vx^* = \arg \underset{\mathbf{0} \leq \vx \leq \mathbf{1}}{\min} L(\vx,\nu^*), 
~\mathbf{1}^T\vx^* = k. 
\end{align}
\end{subequations}
Since the Lagrangian is linearly separable in $\vx$, the first condition simplifies to
\begin{equation}
    x^*_i = \arg \underset{0 \leq x_i \leq 1}{\min} \biggl\{ (\nu^*-d_i)x_i + (\rho/2)(x_i-v_i)^2\biggr\}, \forall \; i \in [n]. 
\end{equation}
The solution of each sub-problem can be computed in closed form as 
\begin{equation}
x_i^* = 
\begin{cases}
0, & v_i < -(1/\rho)(d_i - \nu^*)\\
v_i + (1/\rho)(d_i - \nu^*), & v_i \in [-(1/\rho)(d_i - \nu^*), 1-(1/\rho)(d_i - \nu^*)]\\
1, & v_i > 1-(1/\rho)(d_i - \nu^*)
\end{cases}
\end{equation}
which can be compactly represented as 
\begin{equation}
x_i^* = \max \biggl\{\min{\biggl(v_i + (1/\rho)(d_i - \nu^*),1 \biggr)},0\biggr\}, \forall \; i \in [n].
\end{equation}
\end{proof}

\noindent The above observation suggests a very simple approach to computing $(\vx^*,\nu^*)$. Define the non-linear equation
\begin{equation}
\phi(\nu):= \sum_{i=1}^n \max\biggl\{\min{\biggl(v_i + (1/\rho)(d_i - \nu),1 \biggr)},0\biggr\} - k,
\end{equation}
which is monotone, non-increasing in $\nu$. Since $\phi(\nu^*) = 0$, in order to solve for $\nu^*$ (and hence, $\vx^*$), we can resort to bisection search. We choose the lower and upper limits of the initial bisection interval to be $\nu_{l} := \underset{i \in [n]}{\min}\{d_i + \rho v_i\}-1$ and $\nu_{u}:= \underset{i \in [n]}{\max} \{d_i + \rho v_i\}$ respectively, which yields the initial value interval $[\phi(\nu_{l}),\phi(\nu_{u})]= [n-k,-k]$. Pseudocode for the bisection algorithm is provided in Algorithm \ref{alg:algo1}.

\begin{algorithm}
	\DontPrintSemicolon 
	\SetAlgoLined
	\footnotesize
	\textbf{Input:} $\vv \in \mathbb{R}^n$, degree vector $\vd \in \mathbb{R}^n$, subgraph size $k$, parameter $\rho >0$, exit tolerance $\epsilon > 0$.\;
	\textbf{Output:} The solution $\vx^*:=\textrm{prox}_{\rho g}(\vv)$. \;
	\textbf{Initialize:} $\nu_{l} = \underset{i \in [n]}{\min}\{d_i + \rho v_i\}-1,\nu_{u} = \underset{i \in [n]}{\max} \{d_i + \rho v_i\}$\;
	\Repeat{$\phi(\nu_l) - \phi(\nu_u) \leq \epsilon$}{
		$\nu_m = (\nu_l +\nu_u)/2$\;
		\eIf{$\phi(\nu_m)\phi(\nu_u) < 0$}
		{$\nu_l = \nu_m$\;}
		{$\nu_u = \nu_m$\;}
	}
	\textbf{Return:} $x^*_i = \max \biggl\{\min{\biggl(v_i + (1/\rho)(d_i - \nu_m),1 \biggr)},0\biggr\}, \forall \; i \in [n].$
	\caption{\textsc{Bisection}($\vv,\vd,k,\rho,\epsilon$)}
	\label{alg:algo1}
\end{algorithm}
Note that for a prescribed exit tolerance $\epsilon$, the maximum number of bisection steps is $O(\log[\phi(\nu_l)-\phi(\nu_u)])$, which, for our choice of initial intervals $\{\nu_l,\nu_u\}$, is only $O(\log n)$. Hence, the maximum number of steps required by the bisection algorithm to terminate grows only logarithmically with the problem dimension $n$. We conclude that the above algorithm is an efficient means for evaluating the proximal operator of the function $g$. 

We now turn our attention towards assessing the complexity of computing the proximal operator of the graph total-variation function $h$. Unfortunately, this problem does not admit a simple analytical or computational solution. While its solution can be obtained via solving a sequence of maximum-flow problems \cite{gallo1989fast}, this incurs complexity $O(mn\log(n^2/m))$, which, even for sparse graphs (with $m = O(n))$ is $O(n^2\log n)$. Hence, owing to the high computational complexity of the $\vz$-update, the ADMM framework applied to \eqref{eq:ADMM1} is not scalable to large instances.

In hindsight, the above difficulty appears to stem from the fact that our choice of variable splitting was not effective in yielding simple ADMM updates. Consequently, with the aim of obtaining efficient updates, we introduce a different type of variable splitting. With some abuse of notation, we redefine the function $h$ as
\begin{equation}
h(\vz): = \|\vD\vz\|_1.
\end{equation}
Then, the L-relaxation \eqref{eq:Lrelax} can be equivalently expressed as 
\begin{equation} \label{eq:ADMM2}
 \begin{aligned}
& \underset{\vx \in \mathbb{R}^n, \vz \in \mathbb{R}^m}{\text{min}}
& & g(\vx) + h(\vz)\\
& \quad \text{s.to}& &\vB^T\vx - \vz = \mathbf{0}.
\end{aligned}   
\end{equation}
The ADMM updates for this problem can be shown to be
\begin{subequations}
\begin{align}
\vx^{t+1} &= \arg \underset{\vx}{\min}\; \biggl\{g(\vx) + (\rho/2)\|\vB^T\vx - \vz^t + \vu^t\|_2^2\biggr\}\\
\vz^{t+1} &= \textrm{prox}_{\rho h}(\vB^T\vx^{t+1} + \vu^t) \\
\vu^{t+1} &= \vu^t + \vB^T\vx^{t+1} - \vz^{t+1}
\end{align}
\end{subequations}
where $\vu \in\mathbb{R}^n$ is the normalized dual variable associated with the coupling constraint and $\rho > 0$ is a tuning parameter. Note that in this variant of ADMM, the proximal operator of the function $h$ admits an analytical solution given by \cite[Section 6.5.2]{parikh2014proximal}
\begin{equation}
    \textsc{Shrinkage}(\vv,\vw,\rho): = \max(0,\vv - \vw/\rho) - \max(0,-\vv - \vw/\rho).
\end{equation}
However, the downside is that the simplicity of the $\vx$-update does not carry over from the previous incarnation of ADMM (it is no longer the proximal operator of $g$), which again hinders the scalability of the algorithm. 

The lesson to be learned is that the although the functions $g$ and $h$ have proximal operators which can be evaluated efficiently, the matrix $\vB$ is the ``troublesome'' component as it complicates the primal updates in ADMM, no matter how we elect to perform variable splitting. While this seems like a major drawback of ADMM for our problem, it turns out that there is an \emph{inexact} version of ADMM, which can provide the desired solution. To be precise, we invoke the framework of \emph{Linearized}-ADMM (L-ADMM) \cite{condat2013primal}. In order to motivate the approach, we denote the augmented Lagrangian associated with problem \eqref{eq:ADMM2} as
\begin{equation}\label{eq:AL}
L_{\rho}(\vx,\vz,\vy): = g(\vx) + h(\vz) +\vy^T(\vB^T\vx - \vz) + (\rho/2)\|\vB^T\vx - \vz\|_2^2,
\end{equation}
where $\vy \in \mathbb{R}^m$ is the dual variable corresponding to the coupling constraint. In standard ADMM, the $\vx$-update is computed by minimizing $L_{\rho}(\vx,\vz,\vy)$ with respect to (w.r.t.) $\vx$ while keeping the other variables fixed. In L-ADMM, this update is modified by linearizing the quadratic term in the augmented Lagrangian and adding a new proximal regularization, i.e., replacing $(\rho/2)\|\vB^T\vx - \vz^t\|_2^2$ in \eqref{eq:AL} by 
$$\rho(\vB\vB^T\vx^t - \vB\vz^t)^T\vx + (\mu/2)\|\vx-\vx^t\|_2^2,$$
where $0 < \mu \leq 1/(\rho \|\vB\|_2^2) $ is a regularization parameter. After working out the updates, the algorithm takes the following form
\begin{subequations}
\begin{align}
\vx^{t+1} &= \textrm{prox}_{g/\mu}(\vx^t - \mu\rho\vB(\vB^T\vx^t - \vz^t + \vu^t)) \\
\vz^{t+1} &= \textrm{prox}_{\rho h}(\vB^T\vx^{t+1} + \vu^k) \\
\vu^{t+1} &= \vu^t + \vB^T\vx^{t+1} - \vz^{t+1}.
\end{align}
\end{subequations}
It is evident that L-ADMM accesses both of the functions $g$ and $h$ via their proximal operators only, in contrast to the variants of ADMM considered previously. Hence, each round of ADMM updates can be carried out efficiently, as we have already demonstrated that the proximal operators are easy to compute. We point out that although L-ADMM employs inexact updates, it is still guaranteed to converge to the optimal solution of \eqref{eq:ADMM2}. An even more remarkable feature of L-ADMM is that its convergence does not degrade compared to standard ADMM \cite{he20121}, i.e., it enjoys the same $O(1/t)$ convergence rate. Hence, the L-ADMM algorithm features both lightweight updates and fast convergence. Pseudocode for the algorithm is summarized in Algorithm \ref{alg:algo2}. In practice, we employ an over-relaxation technique \cite[Section 3.4]{boyd2011distributed}, i.e., we replace the term $\vB^T\vx^t$ in the $\vz,\vu$ updates by $\alpha\vB^T\vx^{t+1} + (1-\alpha)\vz^t$, where $\alpha >1 $ is an over-relaxation parameter. We observed that utilizing such a technique improves the empirical convergence of the L-ADMM algorithm.

\begin{algorithm}
	\DontPrintSemicolon 
	\SetAlgoLined
	\footnotesize
	\textbf{Input:} degree vector $\vd \in \mathbb{R}^n$, edge weight vector $\vw \in \mathbb{R}^m$, directed vertex-edge incidence matrix $\vB \in \{-1,0,1\}^{n \times m}$, subgraph size $k$, penalty parameter $\rho >0$, regularization parameter $\mu >0$, over-relaxation parameter $\alpha \in [1.5,1.8]$, bisection exit tolerance $\epsilon > 0$.\;
	\textbf{Output:} A solution of the L-relaxation. \;
	\textbf{Initialize:} $\vx^0 = \textrm{supp}(\textrm{top}_k(\vd)),\; \vz^0 = \vB^T\vx^0, \; \vu^0 = \mathbf{0}$, $\mu = 1/(\rho \|\vB\|_2^2)$, $t \leftarrow 0$\;
	\Repeat{convergence criterion is met}{
		$\vx^{t+1} = \textsc{Bisection}(\vx^t - \mu\rho\vB(\vB^T\vx^t - \vz^t + \vu^t),\vd,k,\rho,\epsilon)$\;
		$\vz^{t+1} = \textsc{Shrinkage}(\alpha\vB^T\vx^{t+1} + (1-\alpha)\vz^t + \vu^t,\vw,\rho)$ \;
		$\vu^{t+1} = \vu^t + \alpha\vB^T\vx^{t+1} + (1-\alpha)\vz^t - \vz^{t+1}$\;
		$t \leftarrow t+1$
	}
	\textbf{Return:} $\vx_L = (1/t){\sum_{i=1}^t \vx^i}$
	\caption{\textsc{L-ADMM}}
	\label{alg:algo2}
\end{algorithm}

Finally, since the solution $\Bar{\vx}$ computed by L-ADMM is not guaranteed to be integral in general, we require a post-processing step into order to ``round'' the solution of the L-relaxation into a binary indicator vector. One such step is to simply project the solution onto the discrete sum-to-$k$ constraints, i.e., we compute
\begin{equation}
    \vx \in \arg\underset{\substack{\vx \in \{0,1\}^n ,\\\mathbf{1}^T\vx = k }}{\min} \|\vx - \vx_L\|_2^2 = \textrm{supp}(\textrm{top}_k(\vx_L))
\end{equation}
which is tantamount to identifying the support of the $k$-largest entries in $\vx_L$, and can be performed in $O(nk)$ time. 

Additionally, we also employ an algorithmic refinement scheme where we use the solution of the L-relaxation to initialize a local-search algorithm. In this scheme, we consider the following  \emph{indefinite} relaxation of the DkS problem
\begin{equation}\label{eq:indef}
    \underset{\vx \in \setP}{\min} \; \biggl\{ f(\vx):=-\vx^T\vW\vx\biggr\}
\end{equation}
which is not convex, and hence cannot be optimally solved in polynomial-time in general. Consequently, we employ the Frank-Wolfe (FW) algorithm \cite{frank1956algorithm} initialized with the solution computed by L-ADMM in order to obtain a high-quality sub-optimal solution.  This is summarized in Algorithm \ref{alg:algo3}. Under the prescribed step-size rule, the algorithm is guaranteed to converge to a stationary point of problem \eqref{eq:indef} \cite[p. 268]{bertsekas2016}. 

\begin{algorithm}[t!]
	\DontPrintSemicolon 
	\SetAlgoLined
	\footnotesize
	\textbf{Input:} Adjacency matrix $\vW \in \mathbb{R}^{n \times n}$, subgraph size $k$, solution of L-ADMM $\vx_L$, Lipschitz constant $L = \|\vW\|_2$.\;
	\textbf{Output:} An approximate solution of the indefinite relaxation \eqref{eq:indef}. \;
	\textbf{Initialize:} $\vx^0 = \vx_L$, $t \leftarrow 0$\;
	\Repeat{convergence criterion is met}{ $\vg^t = -\vW\vx^t$\;
	$\Bar{\vx}^t = \textrm{supp}(\textrm{top}_k(-\vg^t))$\;
	$\alpha^t = \min\{1,((\Bar{\vx}^t - \vx^t)^T\vg^t)/(L\|\Bar{\vx}^t - \vx^t\|_2^2)\}$\;
	$\vx^{t+1} = {\vx}^t + \alpha^t(\Bar{\vx}^t - \vx^t)$\;
	$t \leftarrow t+1$
	}
	\textbf{Return:} $\vx^t$
	\caption{\textsc{Frank-Wolfe}}
	\label{alg:algo3}
\end{algorithm}

\section{Experiments}

In this section, we test the efficacy of the combined L-relaxation and post-processing schemes in discovering $k$-densest subgraphs across a diverse set of real-world graphs. We perform comparisons against a slew of state-of-the-art benchmarks to illustrate the superior performance of our approach. 

\subsection{Datasets}

A summary of the datasets used can be found in Table \ref{tab:stats}, which were retrieved from standard repositories \cite{snapnets,konect}. We pre-processed the datasets (which are unweighted) by symmetrizing the arcs if the network was originally directed, removing all self-loops, and extracting the largest connected component. 

\begin{table}
  \caption{\footnotesize Summary of network statistics: the number of vertices ($n$), the number of edges ($m$), and the network type.}
  \label{tab:stats}
  \footnotesize
  \begin{tabular}{cccl}
    \toprule
    Graph & $n$ & $m$ & Network Type\\
    \midrule
    \textsc{polBlog} & 1,224 & 16,714 & Social\\
    \textsc{Facebook} & 4,039 & 88,234 & Social\\
    \textsc{ppi-Human} & 21,557 & 342K & Biological \\
    \textsc{loc-Gowalla} & 196K & 950K  & Social\\
    \textsc{web-Google} & 875K & 5.10M & Web\\
    \textsc{YouTube} & 1.1M & 2.9M & Social\\
    \textsc{as-Skitter} & 1.7M & 12M & Autonomous Systems\\
    \textsc{wiki-Talk} & 2.4M & 5M & Communications\\
  \bottomrule
\end{tabular}
\end{table}
\subsection{Baselines}

In order to benchmark the performance of our algorithm, we employed the following baselines.
\begin{enumerate}
    \item{\textbf{Greedy:}} The greedy approximation algorithm proposed in \cite[Procedure 2]{feige2001dense}. Given an unweighted graph $\setG$ and a desired subgraph size $k$, the algorithm first constructs a set $\setH$ of the $k/2$ vertices with the largest degree, followed by adding another $k/2$ vertices from $\setV \setminus \setH$ which have the largest number of one-hop neighbors in $\setH$. 
    \item{\textbf{Truncated Power Method (TPM):}} A variant of the classic power method applied to the DkS formulation \eqref{eq:DKS} \cite[Algorithm 2]{yuan2013truncated}. At each step, the algorithm performs standard power-method iterations followed by projecting the result onto the discrete sum-to-$k$ set to ensure iterate feasibility.
    \item{\textbf{Low-rank Binary Principal Component:}} In this approach \cite{papailiopoulos2014finding}, a low rank decomposition of the adjacency matrix $\vW$ is first performed, followed by solving the DkS problem with the low-rank approximation in place of $\vW$. It turns out that in the the rank-$1$ approximation case, the resulting  problem  admits a simple solution in $O(n)$ time, whereas for constant ranks (i.e., $r = O(1)$), instead of checking all $\binom{n}{k}$ possible subsets in the worst-case, the problem can be surprisingly solved in polynomial-time $O(n^{r+1})$. In practice, it is only feasible to run the algorithm for ranks $r \leq 5$, owing to its high complexity. In fact, we were only able to run the algorithm with rank-$1$ approximation for all the datasets considered herein, as even the rank-$2$ case proved too expensive for all but the two smallest datasets.
    \item{\textbf{Edge-density upper bound:}} An important feature of the above approach is that the solution of the DkS problem with rank-$r$ approximation yields an \emph{a posteriori}, data-dependent upper bound on the optimal value of the DkS problem. In formal terms, let $\sigma_1 \geq \sigma_2 \geq \cdots \geq \sigma_d$ denote the $d \leq n$ non-zero singular values of $\vW$. If $\vW_r$ denotes the rank-$r$ approximation of $\vW$, with $\|\vW - \vW_r\|_2 = \sigma_{r+1}$, and $\setS_r^*$ denotes the optimal solution of the rank-$r$ approximation problem for a given $k$, then the quantity
    $$ \min\{1,(\mathbf{1}_{\setS_r^*}^T\vW_r\mathbf{1}_{\setS_r^*}+ \sigma_{r+1})/k-1,\sigma_1/k-1\}$$
    constitutes an upper bound on the edge-density of the optimal $k$-densest subgraph (see \cite[Lemma 3]{papailiopoulos2014finding}). The utility of the above result is that it provides a benchmark for assessing the sub-optimality of a solution generated by \emph{any} algorithm that aims to solve the DkS problem. Although the upper-bound is not attainable in general for every $k$, we demonstrate that the subgraphs computed by our approach can come close to attaining it on real-world graphs for a large range of $k$.
\end{enumerate}

\subsection{Implementation}

We performed all our experiments in Matlab on a Windows workstation equipped with $16$GB RAM and an Intel i7 processor. We used Matlab code for the low-rank principal component approximation approach and TPM \cite{papailiopoulos2014finding}. 

\noindent \textbf{L-ADMM:} Regarding the implementation of our L-ADMM algorithm for solving the L-relaxation, we set the ADMM penalty parameter $\rho = 0.1$, the proximal regularization parameter $\mu = 1/(\rho \|\vB\|_2^2)$, and the over-relaxation parameter $\alpha = 1.8$. The exit tolerance for the bisection subroutine was set to be $\epsilon = 10^{-6}$. The termination criterion of the ADMM algorithm was based on a standard measure \cite[Section 3.3]{boyd2011distributed} - given a pair of absolute and relative tolerances $\epsilon_{\textrm{abs}}$ and $\epsilon_{\textrm{rel}}$ respectively, at each iteration $t$ of ADMM, we compute the primal and dual tolerances  
\begin{subequations}
\begin{align}
  \epsilon_{\textrm{pri}} &= \sqrt{m}\epsilon_{\textrm{abs}}  + \epsilon_{\textrm{rel}}\max\{\|\vB^T\vx^t\|_2,\|\vz^t\|_2\}, \\
   \epsilon_{\textrm{dual}} &= \sqrt{n}\epsilon_{\textrm{abs}}  + \epsilon_{\textrm{rel}}\|\vB\vu^t\|_2.
\end{align}
\end{subequations}
Defining the primal and dual residuals  $\vr^t: = \vB^T\vx^t - \vz^t$ and $s^t: = \vB(\vz^t-\vz^{t-1})$ respectively, we stop the algorithm when these residuals are small in the sense that $\|r^t\|_2 \leq \epsilon_{\textrm{pri}}$ and $\|\vs^t\|_2 \leq \epsilon_{\textrm{dual}}$, or a maximum of $3000$ iterations have been performed. In our experiments, we set $\epsilon_{\textrm{abs}} = \epsilon_{\textrm{rel}} = 10^{-3}$ for all datasets excepting $\textsc{web-Google}$ and $\textsc{YouTube}$, for which we used the setting $\epsilon_{\textrm{abs}} = \epsilon_{\textrm{rel}} = 10^{-4}$.

\noindent \textbf{FW and TPM:} We initialized both algorithms with the solution returned by the L-ADMM algorithm. Note that for TPM, the solution of the L-relaxation is a superior initialization compared to selecting the support of the $k$-vertices with the largest degree (originally proposed in \cite{yuan2013truncated}), i.e., here we give TPM the benefit of the doubt. The algorithms are run till they attain convergence in terms of the cost function, or a maximum of 100 iterations are reached. Finally, while the solution of FW is not guaranteed to be integral in general, we observed in our experiments that the algorithm returns a solution that is integral (up to machine precision), and thus we did not perform a rounding step at the end.   
    \begin{figure}[t!]
    \includegraphics[width = 0.23\textwidth]{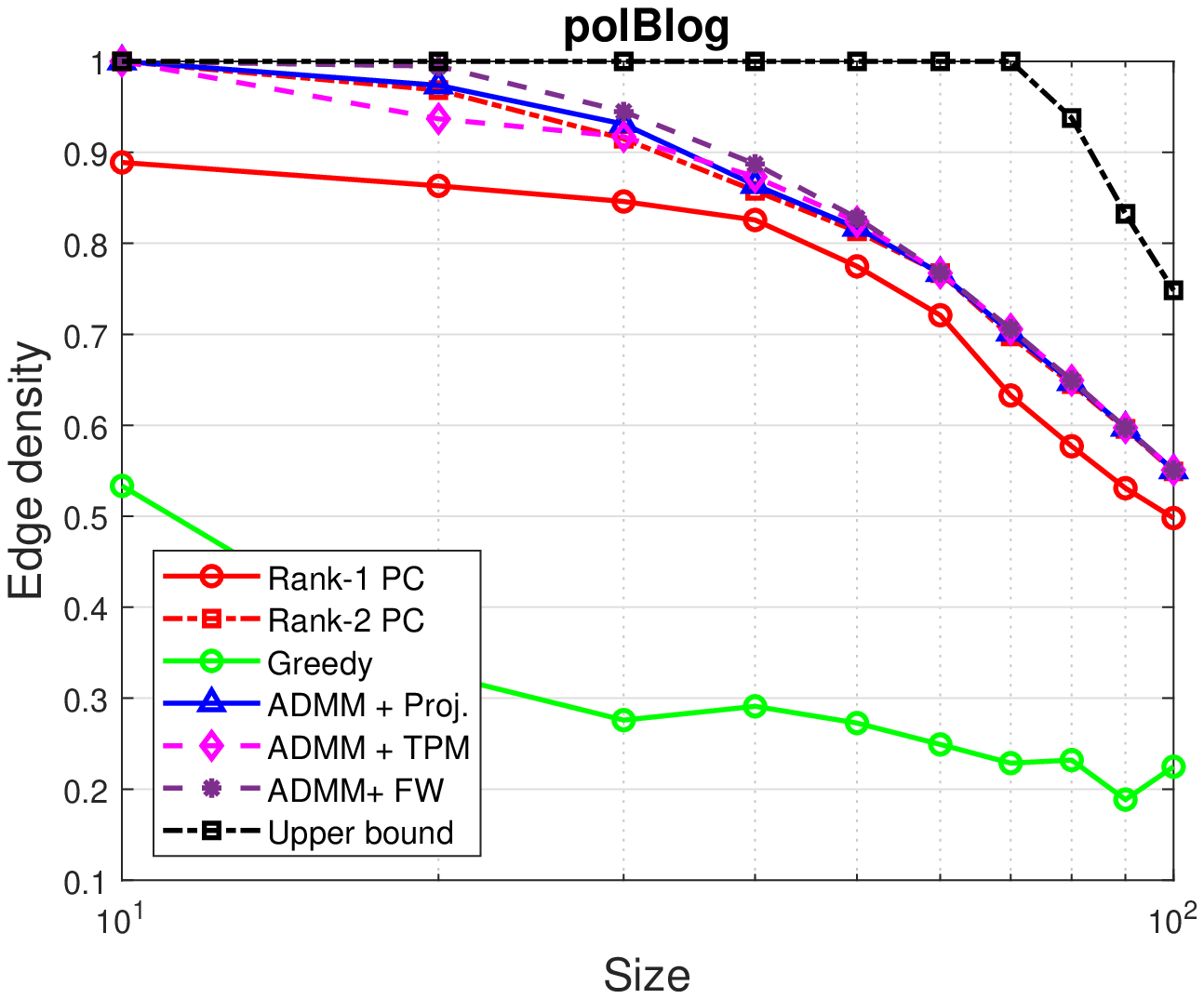}
    \includegraphics[width = 0.23\textwidth]{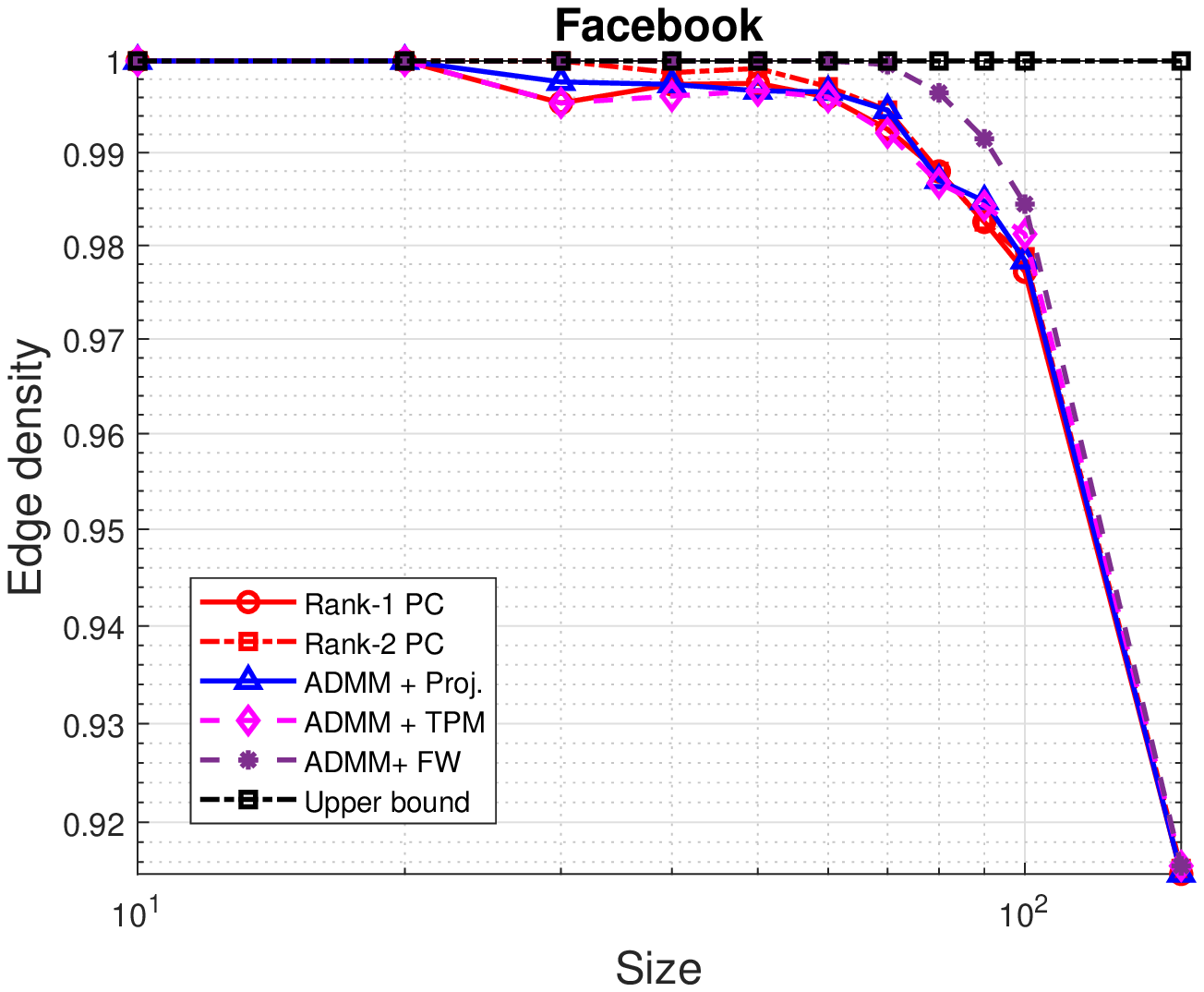}
    \includegraphics[width = 0.23\textwidth]{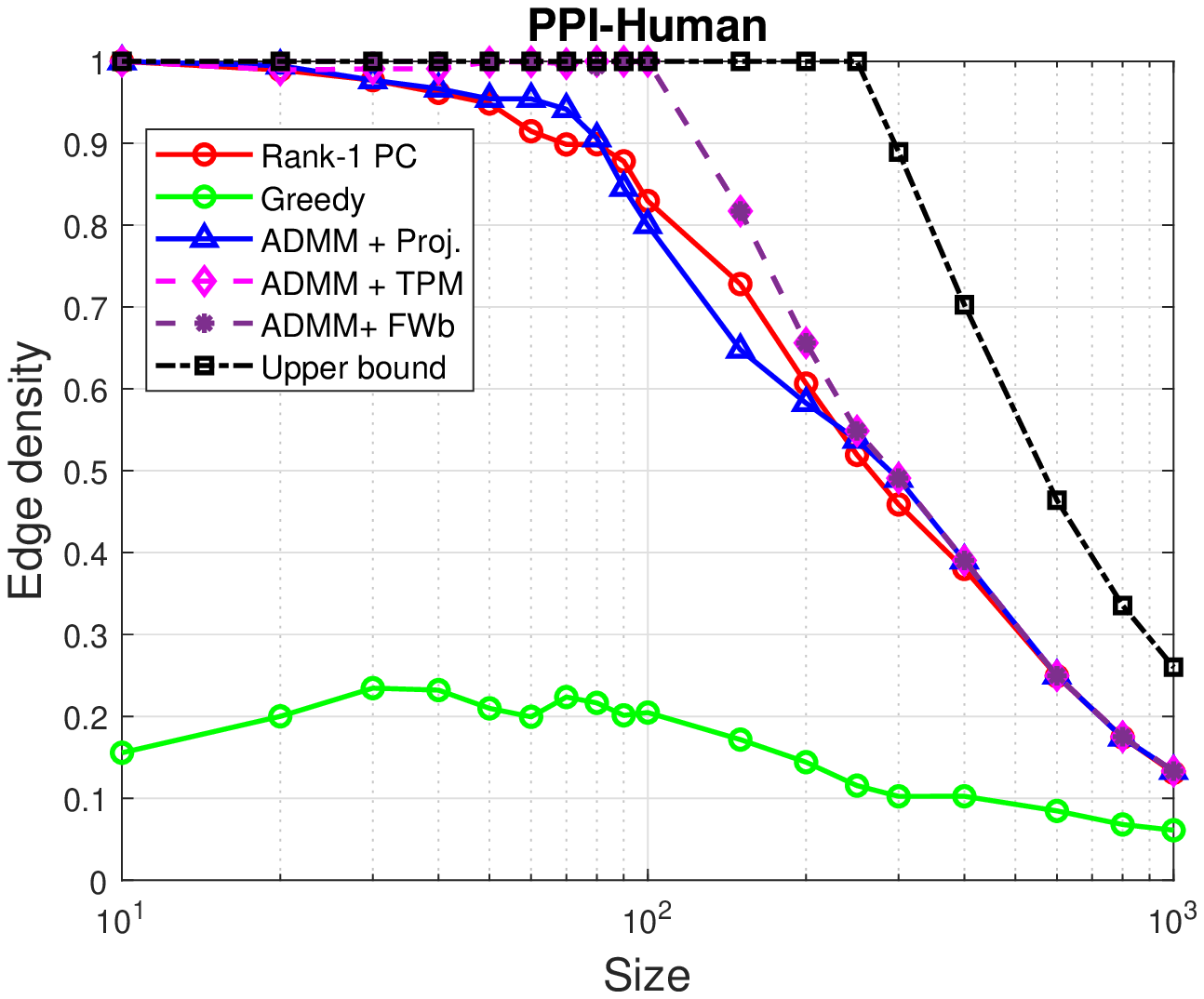}
    \includegraphics[width = 0.23\textwidth]{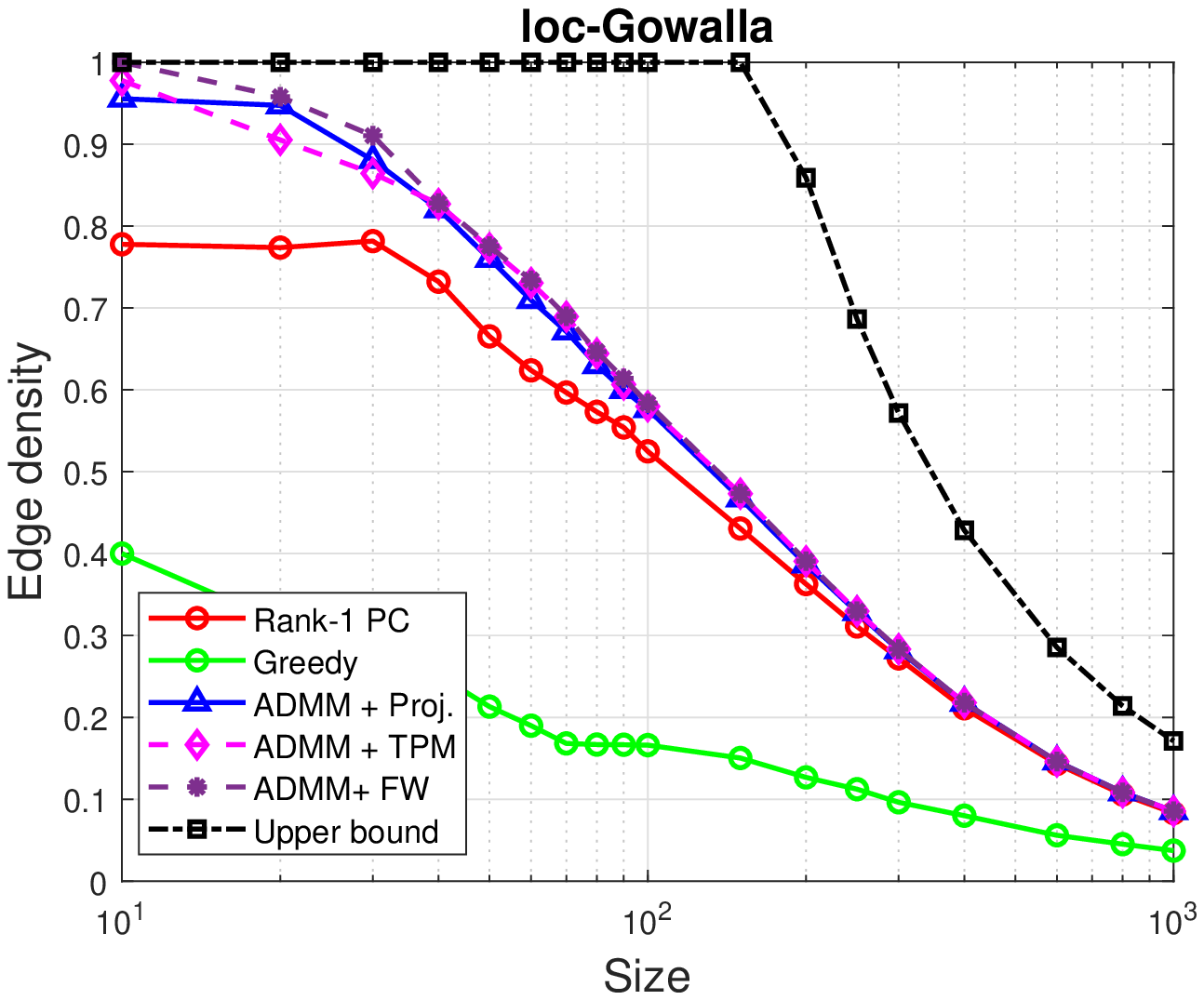}
    \includegraphics[width = 0.23\textwidth]{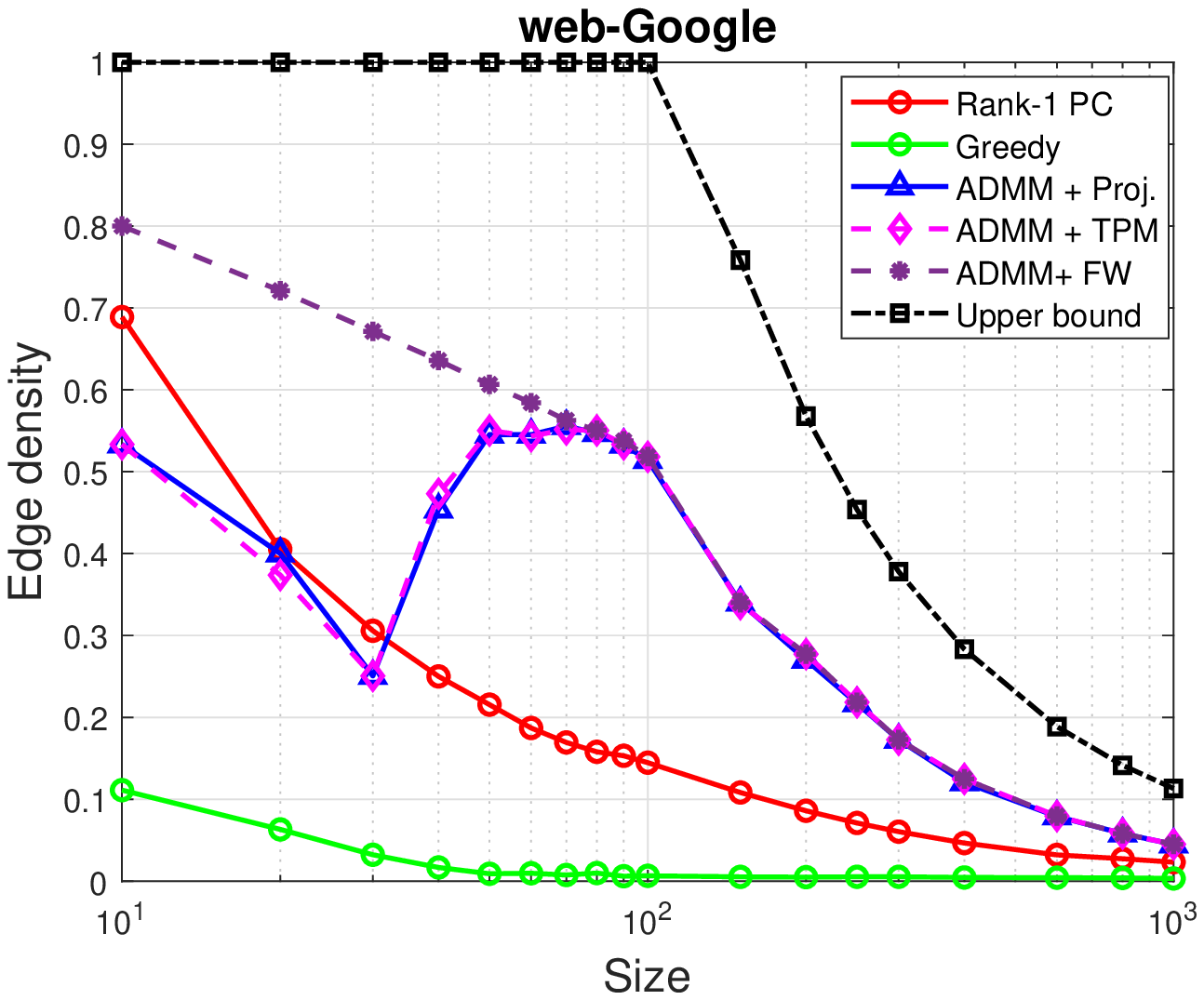}
    \includegraphics[width = 0.23\textwidth]{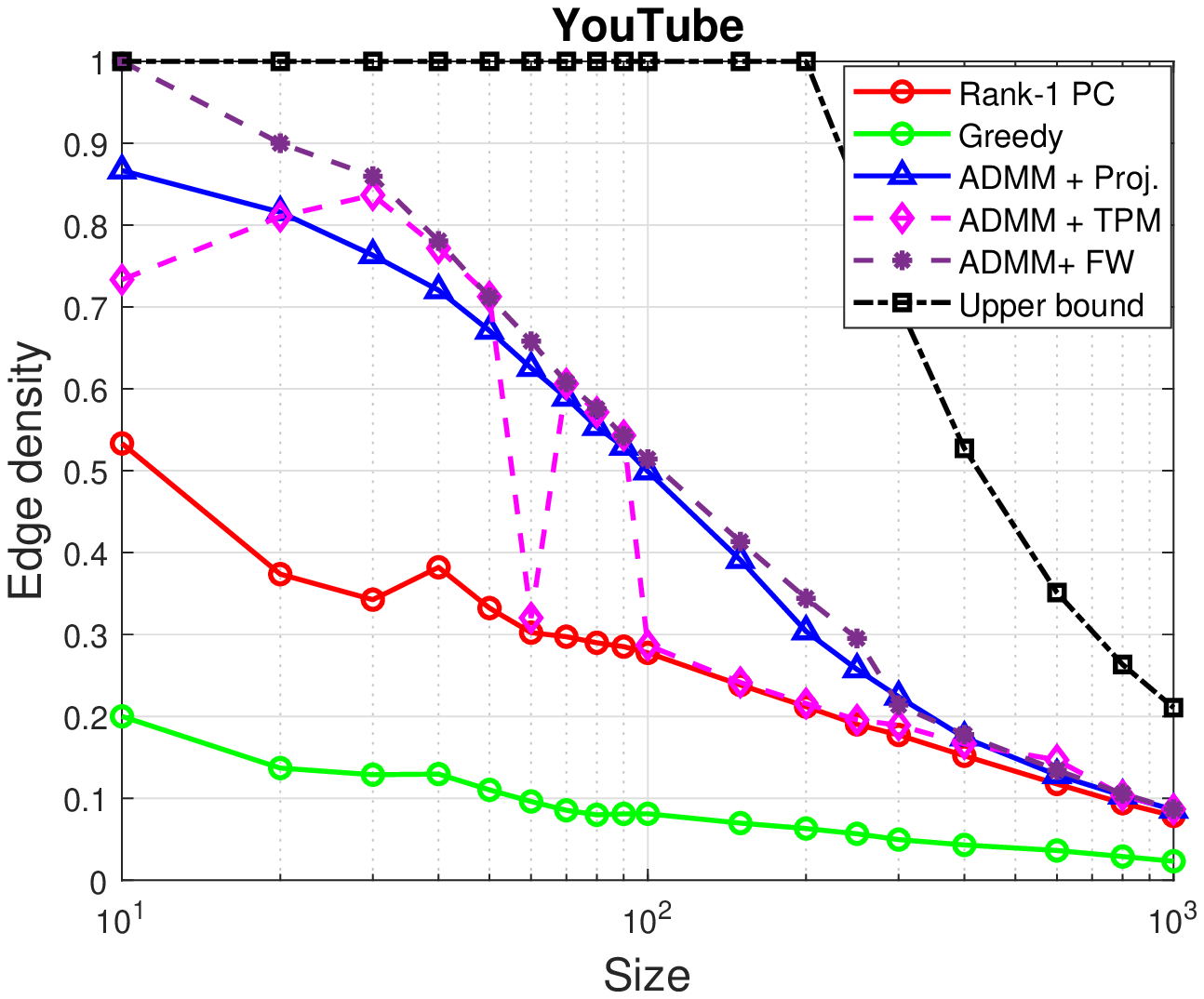}
    \includegraphics[width = 0.23\textwidth]{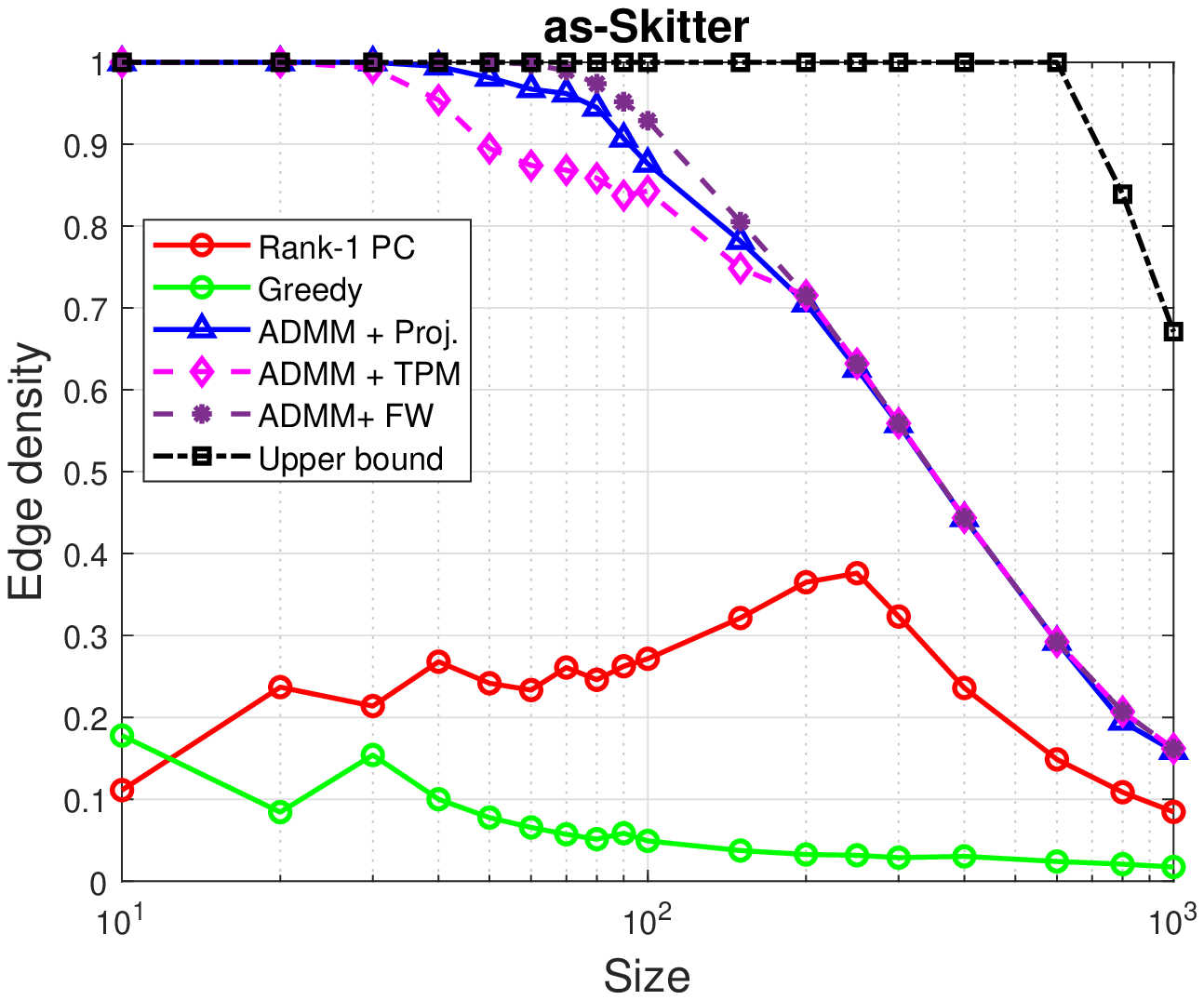}
    \includegraphics[width = 0.23\textwidth]{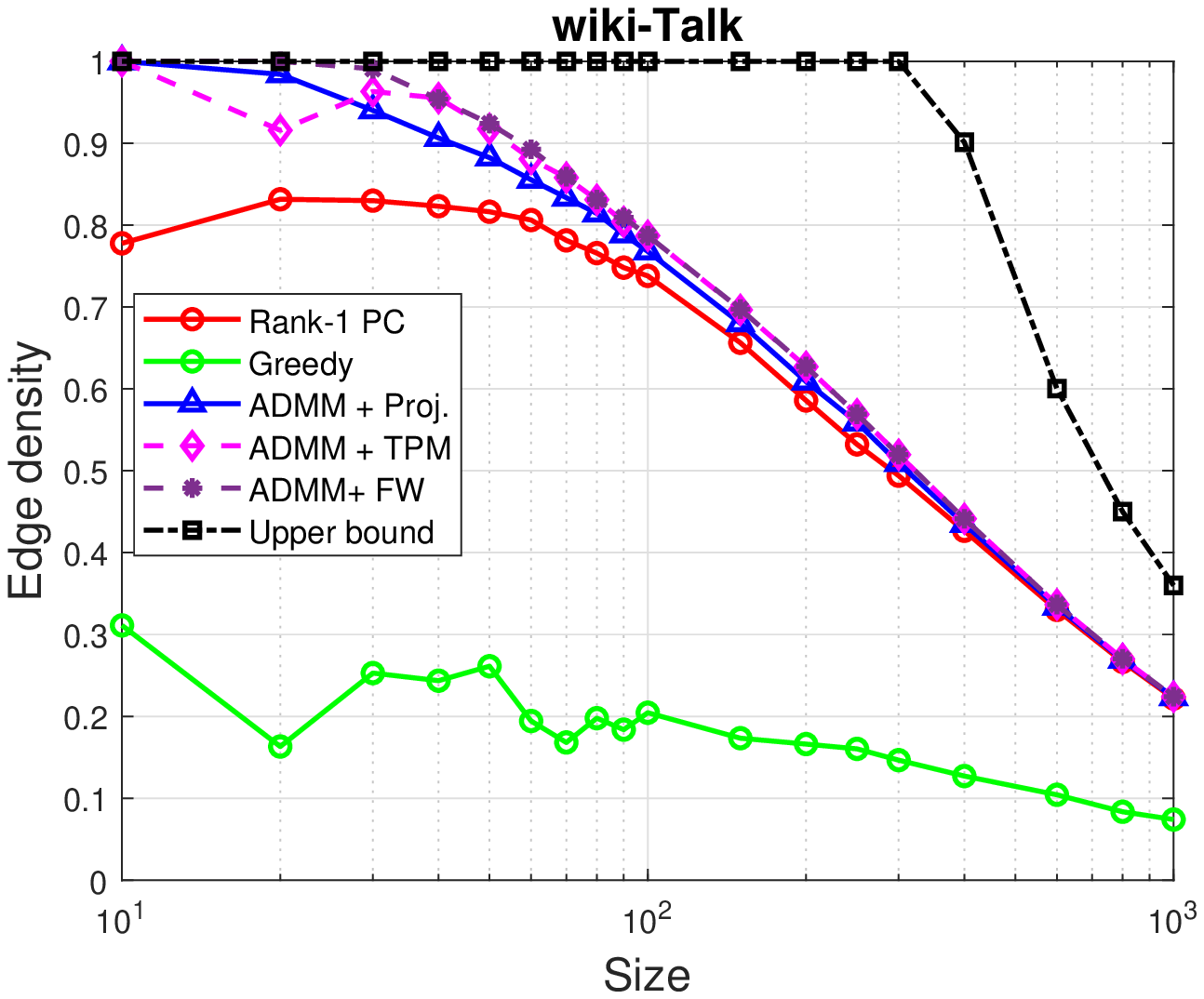}
    \caption{ Edge density vs size: We ran the rank-2 approximation only on the 2 smallest datasets owing to its complexity. The greedy algorithm is omitted from comparison on \textsc{Facebook} owing to its poor performance relative to the other baselines. For \textsc{Facebook} and \textsc{polBlog}, the upper-bound is computed w.r.t. the rank-2 approximation, while it is w.r.t. the rank-1 approximation on the remaining datasets.}
    \label{fig:density}
\end{figure}

\subsection{Results}
The outcomes of our experiments are depicted in Figure \ref{fig:density} and \ref{fig:time}, which depict the edge density of the subgraphs determined by the methods and the runtimes versus subgraph size $k$, respectively. Our main findings are as follows:

\begin{itemize}
    \item The upper-bound on the optimal edge-density computed from solving the low-rank approximation to the DkS problem is very useful in gauging the sub-optimality of the solutions computed by the different methods. It reveals that in contrast to pessimistic worst-case results regarding the DkS problem, several methods (with the exception of the greedy algorithm) can yield high quality solutions on real-world graphs.
    
    \item Our proposed approach, the L-relaxation coupled with the two rounding techniques (projection and iterative refinement via the Frank-Wolfe algorithm) performs very well. In particular, the latter scheme is consistently the best overall, outperforming TPM, the low-rank approximation, and the solution obtained by projecting the L-ADMM solution. Although TPM shares the same initialization as FW, it can exhibit non-monotone behavior with regard to density as the size is varied. We attribute this to the fact that FW is guaranteed to converge to a stationary point of the indefinite relaxation (which is empirically observed to be integral), whereas TPM simply increases the objective function of DkS. Furthermore, for small values of $k \leq 100$ (the regime where one intuitively expects the densest subgraphs to be present), L-ADMM + FW can attain the upper-bound in many cases, which is clearly optimal; otherwise it attains the most significant fraction of the upper-bound (typically $65-80\%$ for $k \leq 100$).
    
    \item The runtime of the rank-$2$ approximation algorithm scales unfavorably relative to the other methods, and hence it is omitted from the larger datasets. While ADMM comes second in terms of complexity, it is by no means unaffordable, taking an average of $15$ minutes to terminate on the largest graphs. This is due to the simplicity of its subroutines, which require performing bisection search and shrinkage at each step. Additionally, compared to running ADMM, the complexity of performing iterative refinement via the Frank-Wolfe algorithm is substantially  smaller. 
\end{itemize}
Our investigation reveals that solving the L-relaxation via L-ADMM followed by refining the solution via few iterations of FW constitutes a potent and efficient algorithmic framework for effectively mining dense subgraphs from real-world graphs.

\begin{figure}[t!]
    \includegraphics[width = 0.23\textwidth]{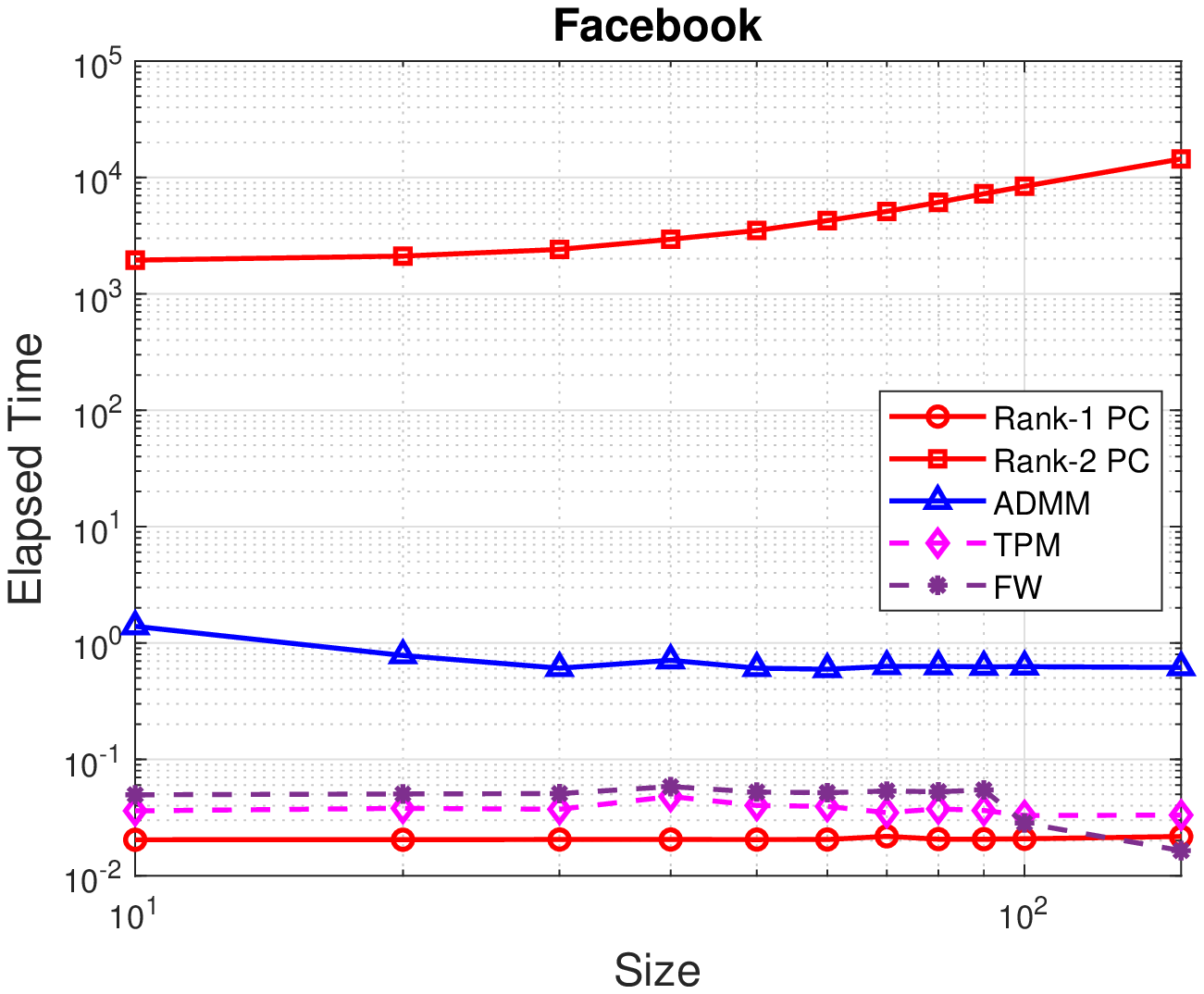}
    \includegraphics[width = 0.23\textwidth]{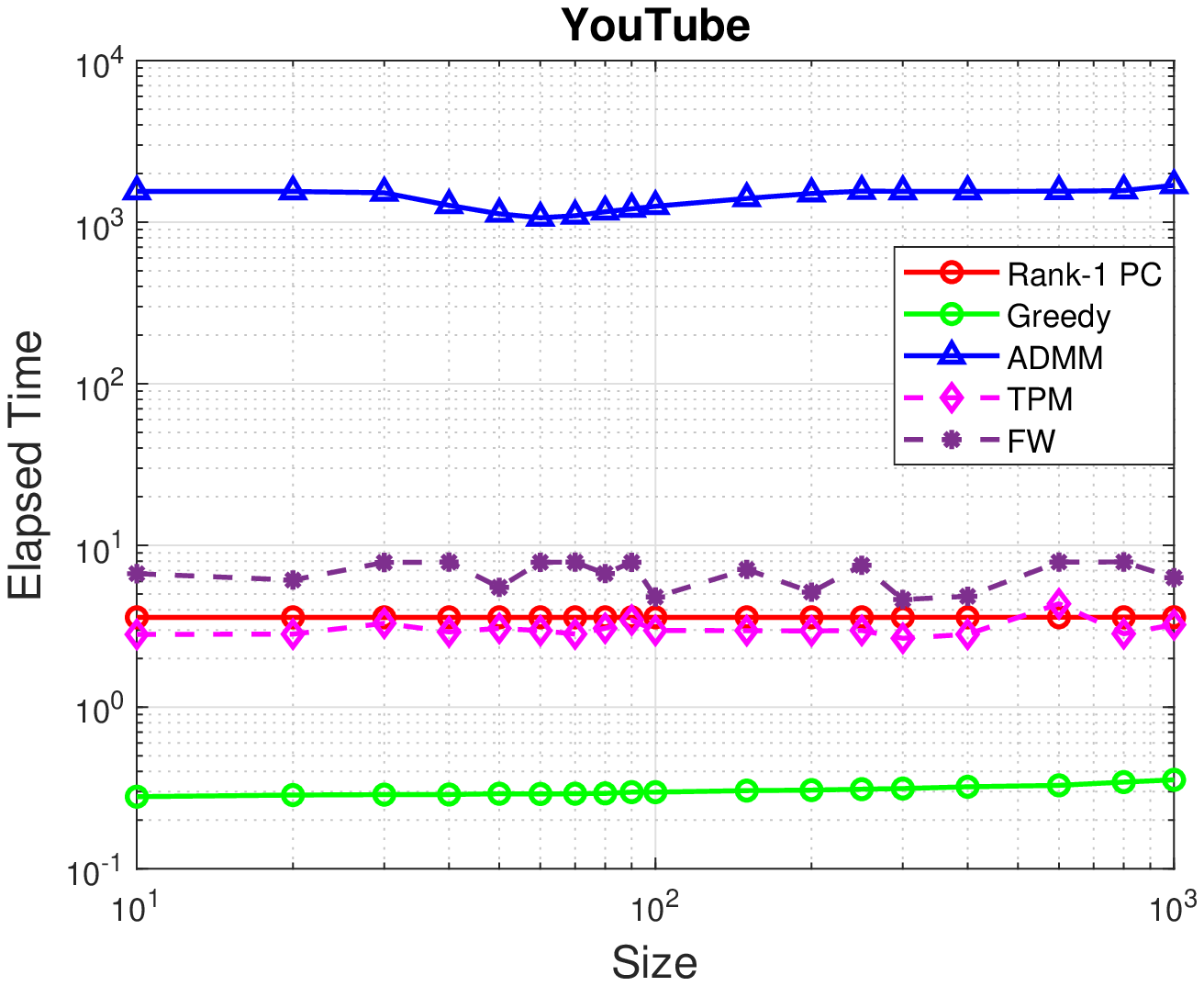}
    \includegraphics[width = 0.23\textwidth]{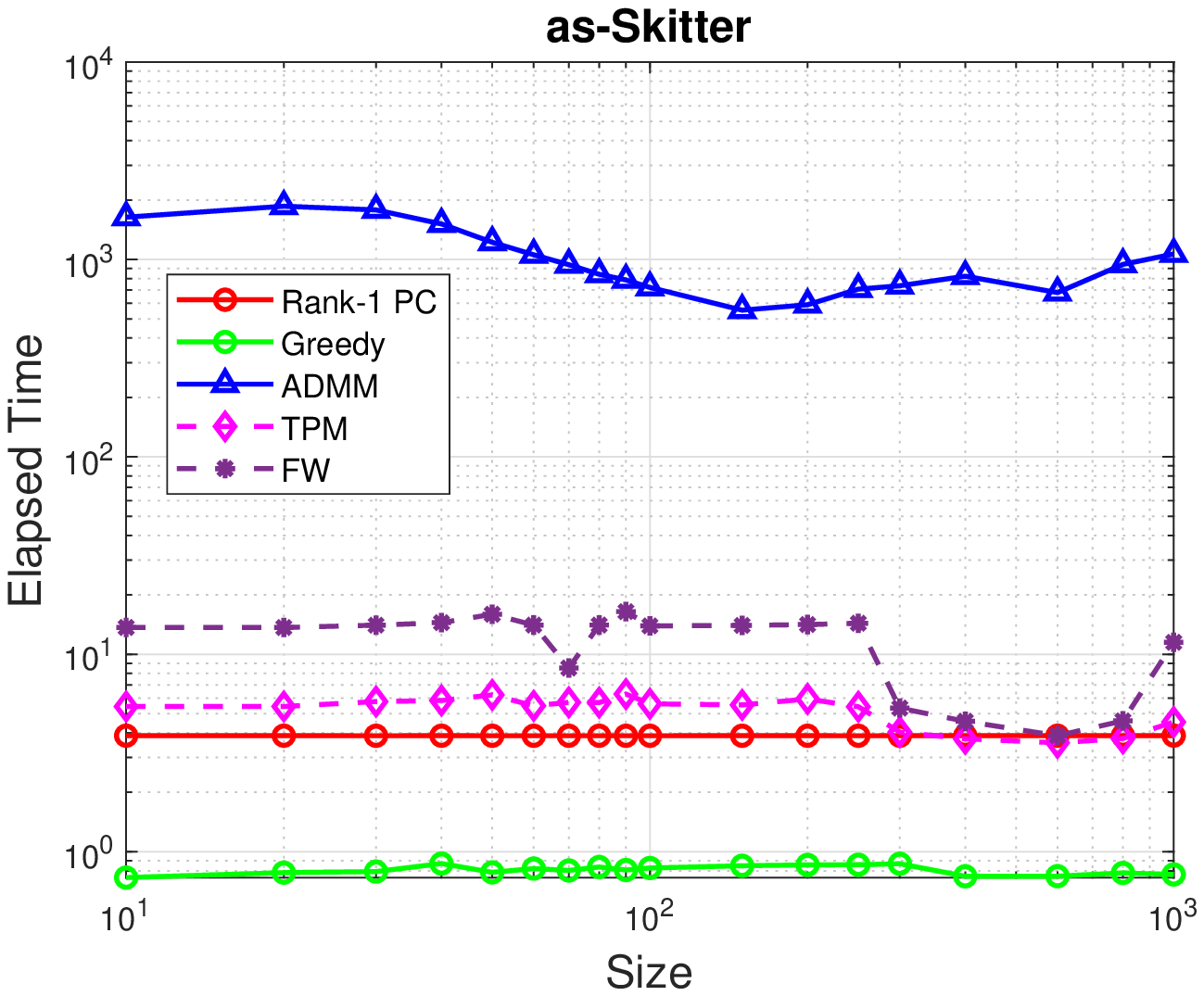}
    \includegraphics[width = 0.23\textwidth]{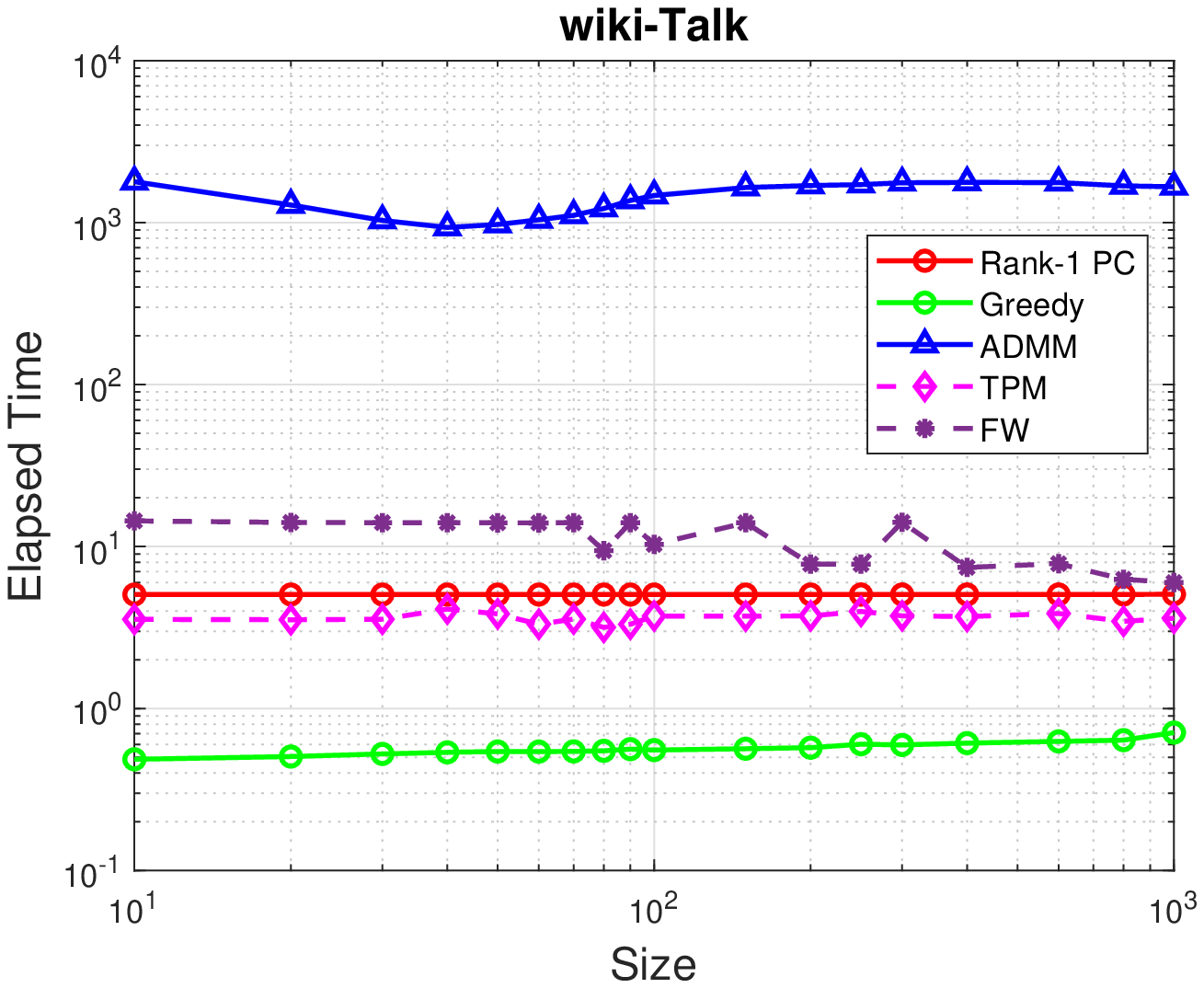}
    \caption{Runtime vs size on selected, representative datasets, owing to space constraints.}
    \label{fig:time}
\end{figure}

\section{Conclusion}

We considered the \textsc{Densest}-$k$-\textsc{Subgraph} problem (DkS), and reformulated it as minimizing a submodular cost function subject to a cardinality constraint. Adopting this viewpoint, we proposed a convex relaxation of DkS that minimizes the Lov\'asz extension of the submodular cost function over the convex hull of the cardinality constraint. While the Lov\'asz extension does not admit a closed form expression in general, we showed that for DkS it does admit an analytical form. We exploited this form to develop an efficient algorithm based on an inexact variant of the Alternating Direction Method of Multipliers (ADMM) that is capable of solving the relaxed problem at scale. After rounding the solution returned by ADMM via the proposed schemes, we conducted experiments on real-world graphs to showcase the effectiveness of our approach compared to prevailing baselines. Contrary to pessimistic worst-case results, our relaxation scheme is very effective at exploring the edge-density vs size curve in real-world graphs, yielding subgraphs that are no worse than $65-80\%$ of the optimal density. 

\section{Acknowledgements}
Supported by the National Science Foundation and the Army Research Office under Grants No. IIS-1908070 and ARO-W911NF1910407 respectively.

\bibliographystyle{ACM-Reference-Format}
\bibliography{sample-base}

\end{document}